\documentclass{ws-ijcga}
\usepackage[utf8]{inputenc}
\usepackage{graphicx}
\usepackage{cite}
\usepackage{verbatim}
\usepackage{amsmath,amssymb}
\usepackage{color}
\usepackage{wrapfig}

\title{Similarity of closed polygonal curves in Frechet metric}

\author{Schlesinger M.I., Vodolazskiy E.V., Yakovenko V.M.}

\begin{document}

\markboth{Schlesinger M.I., Vodolazskiy E.V., Yakovenko V.M.}{Similarity of closed polygonal curves in Frechet metric}

\catchline

\author{ Schlesinger M.I.}

\address{
schles@irtc.org.ua
}

\author{Vodolazskiy E.V.}

\address{
    waterlaz@gmail.com
}

\author{Yakovenko V.M.}

\address{
    asacynloki@gmail.com
}

\address{
International Research and Training Centre \\
of Information Technologies and Systems \\
National Academy of Science of Ukraine \\
Cybernetica Centre, \\
prospect Academica Glushkova, 40, \\
03680, Kiev-680, GSP, Ukraine. 
}

\maketitle

\pub{Received (received date)}{Revised (revised date)}
{Communicated by (Name)}

\begin{abstract}
The article analyzes similarity of closed polygonal curves in Frechet metric, 
which is stronger than the well-known Hausdorff metric and therefore is more appropriate in some applications.
An algorithm that determines whether the Frechet distance between two closed polygonal curves with $m$ and $n$ vertices is less than a given number~$\varepsilon$ is described.
The described algorithm takes $O(m n)$ time whereas the previously known algorithms take $O(m n\log (m n))$ time.  

\end{abstract}
\section{Introduction}
Frechet metric is a tool for 
cyclic process analysis and image processing \cite{ICDAR.2007.121} as well as the well-known Hausdorff metric \cite{Chen02}\cite{rockafellar2011variational}
\cite{ismm2011}. 
The Frechet metric is stronger than the Hausdorff metric and therefore is more appropriate in some applications \cite{frechet}. 
The Frechet metric for closed polygonal curves has been studied in a paper \cite{frechet} by Alt and Godau.
They propose an algorithm that determines whether the distance between two closed polygonal curves with $m$ and $n$ vertices is greater than a given number~$\varepsilon$.
The complexity of the algorithm is $O(m n \log (m n))$ on a random access machine that performs arithmetical operations 
and computes square roots in constant time. 
Our paper develops the ideas of the original paper \cite{frechet} and solves the same problem in $O(m n)$ time.
Sections \ref{Definitions} and \ref{Diagrams} describe the concepts that are common for both papers. 
The end of Section \ref{Diagrams} shows the difference between the known and the proposed approaches.
Sections \ref{Achievability}-\ref{ResultSection} explicate the proposed approach.

\section{\label{Definitions} Problem definition.}
Let $\mathbb{R}^k$ be a $k$-dimensional linear space with a metric $\text{   }d:\mathbb{R}^k\times \mathbb{R}^k \rightarrow \mathbb{R}$, 
where  $d(x,y)=\sqrt{(x-y)^2}$ is the distance between points $x, y \in \mathbb{R}^k$.

\begin{definition}
A closed polygonal curve $X$ with $m$ vertices is a sequence $\bar{x}=(x_0, x_1, \cdots, x_m=x_0)$, $x_i \in \mathbb{R}^k$, and a function
$f_X: \{ t \in \mathbb{R}|0 \le t \le m \} \rightarrow \mathbb{R}^k$ such that 
$f_X(i+\alpha)=(1-\alpha) x_i + \alpha x_{i+1}$
for $i \in \{0, 1, \dots , m-1\}$ and $0 \le \alpha \le 1$.
\end{definition}
\begin{definition}
A cyclic shift of an interval $\{ t|0 \le t \le m \}$ by a value $\tau$, $0 \le \tau \le m $, is a function
$s:\{ t|0 \le t \le m \} \rightarrow \{ t|0 \le t \le m \}$ that depends on a parameter $\tau$ such that 
$s(t;\tau) =t+\tau$ for $t+\tau \le m$ and $s(t;\tau) =t+\tau - m$ for $t+\tau > m$. 
\end{definition}
For any number $m$ denote $W_m$ a set of all monotonically non-decreasing continuous functions
$w:\{ t|0 \le t \le 1 \}\rightarrow\{ t|0 \le t \le m \}$ such that $w(0)=0,w(1)=m$. 
\begin{definition}
A function $\varphi:\{t|0 \le t \le 1\}\rightarrow \mathbb{R}^k$ is called a monotonic traverse of a closed $m$-gonal curve $X$ 
if there is a function $w \in W_m$ and a number $\tau$, $0 \le \tau \le m$, such that $\varphi(t)=f_X(s(w(t);\tau))$ for all $t$, $0 \le t \le 1$.
%$s$ is a cyclic shift of an interval $\{ t|0 \le t \le m \}$.
\end{definition}
For given closed polygonal curves $X$ and $Y$ denote $\Phi_X$ and $\Phi_Y$ sets of their monotonic traverses.
\begin{definition}
Frechet distance between closed polygonal curves $X$ and $Y$ is a number
\begin{equation} \nonumber
\delta(X,Y) = \min_{\varphi_X \in \Phi_X} \min_ {\varphi_Y \in \Phi_Y } \max_{0\le t \le 1}d(\varphi_X(t) , \varphi_Y(t)).
\end{equation}
\end{definition}
The article solves the problem of developing an algorithm that determines whether $\delta(X,Y) \le \varepsilon$
for given closed polygonal curves $X$ and $Y$ and a number~$\varepsilon$.

\section{\label{Diagrams}The free space diagram.} 
The problem's analysis is largely based on its representation in a form of a free space diagram introduced in a paper \cite{frechet}. 
For two numbers $m$ and $n$ let us define a rectangle $\widetilde{D} = \{(u,v) \in \mathbb{R}^2|0 \le u \le m, 0 \le v \le n \}$, 
$u$ and $v$ being horizontal and vertical coordinates of a point $(u,v)$.
For two closed polygonal curves $X$ and $Y$ with $m$ and $n$ vertices and a number $\varepsilon$ a subset
$\widetilde{D}_\varepsilon= \{(u,v) \in  \widetilde{D}| d(f_X(u),f_Y(v)) \le \varepsilon\}$ is defined. 
Let us also define a rectangle
$D = \widetilde{D} \cup \{(u+m,v)|(u,v) \in \widetilde{D}\}$ with its subset 
$D_\varepsilon = \widetilde{D}_\varepsilon \cup \{(u+m,v)|(u,v) \in \widetilde{D}_\varepsilon\}$ called a free space.
Denote $T$, $B$, $L$ and $R$
the top, bottom, left and right sides of the rectangle $D$.
Two closed polygonal curves $X$ and $Y$ are shown on Figure \ref{Fig0} and the corresponding rectangles
$\widetilde{D}$ and  $D$ are shown on Figure \ref{Fig1}, the light region of $D$ being the free space $D_\varepsilon$. %\\
\begin{figure}
  \centering
  \includegraphics*{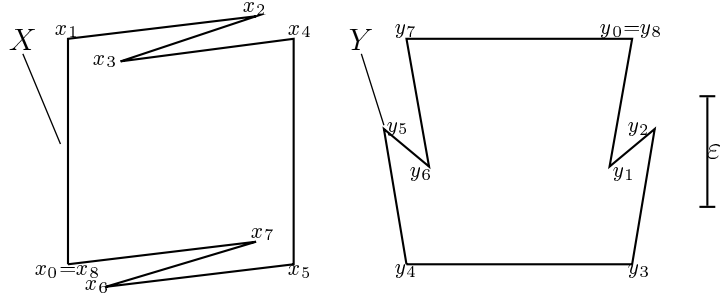}
  \caption{Two closed polygonal curves}
  \label{Fig0}
\end{figure}
\begin{figure}
  \centering
  \includegraphics*[width=\textwidth]{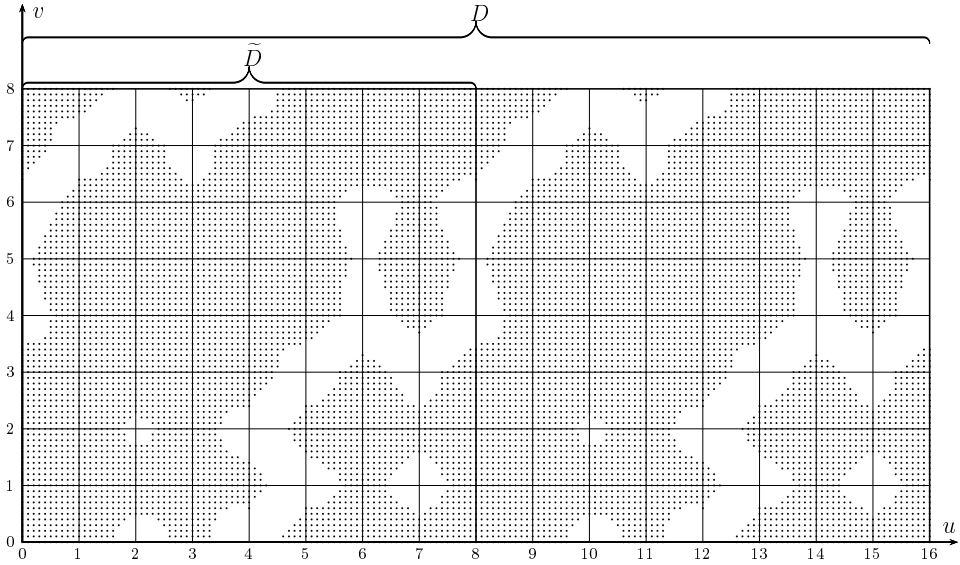}
  \caption{The free space diagram}
  \label{Fig1}
\end{figure}
%\\
\begin{definition}
A monotonic path is a connected subset $\gamma \subset D_\varepsilon$ such that $(u~-~u')(v~-~v') \ge 0$ for any two points 
$(u,v) \in \gamma$, $(u',v') \in \gamma$.
\end{definition}
\begin{definition}
Two points $(u,v) \in D$ and $(u',v') \in D$ are mutually reachable if and only if a monotonic path $\gamma$ exists such that $(u,v) \in \gamma$, $(u',v') \in \gamma$.
\end{definition}
Mutual reachability defines a subset of $D_\varepsilon \times D_\varepsilon$ and is a symmetric reflexive binary relation, not necessarily transitive. 
The expression "a point $(u,v) \in D$ is reachable from $(u',v') \in D$" is further used in a sense that $(u,v)$ and $(u',v')$ are mutually reachable.
\begin{definition}
A point $(u,v) \in D_\varepsilon$ is reachable from the bottom if it is reachable from at least one point from $B$;
a point $(u,v) \in D_\varepsilon$ is reachable from the top if it is reachable from at least one point of $T$.
\end{definition}
Reachability from top and bottom are unary relations. Each of these two relations define a subset of $D_\varepsilon$.
Denote $g_\downarrow \subset D_\varepsilon$ a set of points reachable from the bottom and
$g^\uparrow \subset D_\varepsilon$ a set of points reachable from the top.
\begin{figure}
  \centering
  \includegraphics*[width=\textwidth]{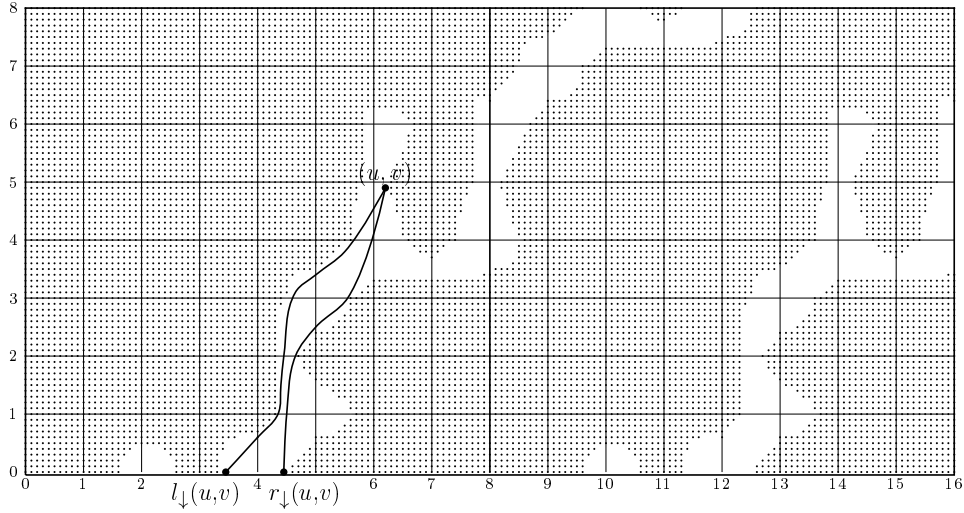}
  \caption{The set of points reachable from the bottom}
  \label{Fig2}
\end{figure}
\begin{figure}
  \centering
  \includegraphics*[width=\textwidth]{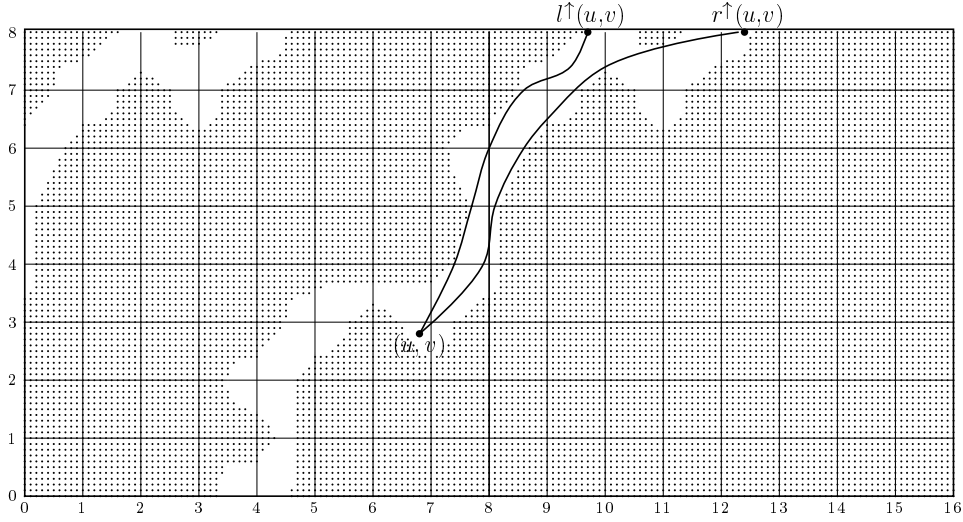}
  \caption{The set of points reachable from the top}
  \label{Fig3}
\end{figure}
Let us define four functions\\
\phantom{}\quad $l_\downarrow:g_\downarrow \rightarrow \{u|0 \le u \le 2m \} $, \quad $r_\downarrow:g_\downarrow \rightarrow \{u|0 \le u \le 2m \} $, \\
\phantom{}\quad $l^\uparrow:g^\uparrow \rightarrow \{u|0 \le u \le 2m \}, \quad $ $r^\uparrow:g^\uparrow \rightarrow \{u|0 \le u \le 2m \} $
\\
such that $l_\downarrow(u,v)$ (or $r_\downarrow(u,v)$) is the minimum (or maximum) value $u^*$ such that $(u,v)$ is reachable from $(u^*,0) \in B$.
Similarly, $l^\uparrow(u,v)$ (or $r^\uparrow(u,v)$) is the minimum (or maximum) value $u^*$, such that $(u,v)$ is reachable from $(u^*,n) \in T$.
The subset $g_\downarrow \subset D_\varepsilon$ is shown by a light area on Figure \ref{Fig2}. 
The same Figure \ref{Fig2} shows the values $l_\downarrow(u,v)$ and $r_\downarrow(u,v)$) for some point $(u,v) \in g_\downarrow$. 
Figure \ref{Fig3} shows the subset $g^\uparrow \subset D_\varepsilon$ and values $l^\uparrow(u,v)$ and $r^\uparrow(u,v)$.
The following two lemmas are proved in the paper \cite{frechet} by Alt and Godau (lemma 9 and 10 in \cite{frechet}).
\begin{lemma} \label{AltFirst}
The distance between closed polygonal curves $X$ and $Y$  is not greater than $\varepsilon$ if and only if there exists a number 
$u$, $0 \le u \le m$, such that
the points $(u + m, n)$ and $(u,0)$ are mutually reachable.
\end{lemma} 
\begin{lemma} \label{AltSecond}
Two points $(u^{top},n) \in T$ and $(u_{bot},0) \in B$ are mutually reachable if and only if 
$(u^{top},n) \in g_\downarrow $, $(u_{bot},0) \in g^\uparrow$ and $l^\uparrow(u_{bot},0) \le u^{top} \le r^\uparrow(u_{bot},0)$.
\end{lemma} 
According to these two lemmas testing the inequality $\delta(X,Y) \le \varepsilon$ is reduced in the paper \cite{frechet} to determining if there exists a value $u$ that fulfills
\begin{equation} \label{AltDecision}
(u,0) \in g^\uparrow, \quad (u+m,n) \in g_\downarrow, \quad l^\uparrow(u,0) \le u+m \le r^\uparrow(u,0).  
\end{equation}
%The paper \cite{frechet} describes an algorithm that practically represents the data needed to test conditions (\ref{AltDecision}). 
%The complexity of that algorithm is $O(m  n\log (m n))$.
%\\
Our algorithm is also based on Lemma \ref{AltFirst}. 
However, instead of Lemma \ref{AltSecond} we rely on a similar Lemma \ref{OurSecond}.
We prove Lemma \ref{OurSecond} in the same manner as Lemma \ref{AltSecond} has been proved in \cite{frechet}.      

\begin{lemma} \label{OurSecond}
Two points $(u^{top},n) \in T$ and $(u_{bot},0) \in B$ are mutually reachable if and only if\\ 
\phantom{} \quad\quad $(u^{top},n) \in g_\downarrow $, \quad $(u_{bot},0) \in g^\uparrow$, \quad $u^{top} \le r^\uparrow(u_{bot},0)$, 
\quad $u_{bot} \le r_\downarrow(u^{top},n)$. 
\end{lemma} 
\begin{proof}
Obviously, if $(u^{top},n) \in T$ and $(u_{bot},0) \in B$ are mutually reachable then 
$(u^{top},n) \in g_\downarrow $, $(u_{bot},0) \in g^\uparrow$ and $u^{top} \le r^\uparrow(u_{bot},0)$, $u_{bot} \le r_\downarrow(u^{top},n)$.
The reverse implication is also valid, which is illustrated by Figure \ref{Fig5}.\\
\begin{figure}
  \centering
  \includegraphics{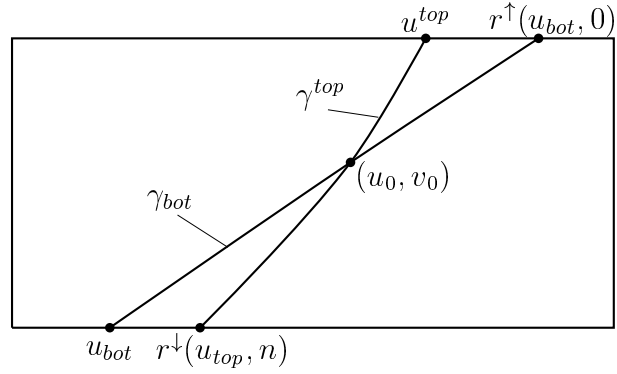}
  \caption{There is a path between $(u_{bot}, 0)$ and $(u^{top},n)$}
  \label{Fig5}
\end{figure}
It follows from condition $(u^{top},n) \in g_\downarrow $ that a monotonic path $\gamma^{top} \subset D_\varepsilon$
connecting $(u^{top},n)$ and $(r_\downarrow(u^{top},n),0)$ exists.
Similarly, it follows from condition $(u_{bot},0) \in g^\uparrow$ 
that a monotonic path $\gamma_{bot} \subset D_\varepsilon$ connecting $(u_{bot},0)$ and $(r^\uparrow(u_{bot},0),n)$ exists. 
Both paths are connected subsets and therefore they intersect in at least one point $(u_0, v_0)$ according to conditions $u^{top} \le r^\uparrow(u_{bot},0)$, $u_{bot} \le r_\downarrow(u^{top},n)$.
Let us build a path $\gamma$ that consists of a part of the path $\gamma_{bot}$ from $(u_{bot},0)$ to $(u_0, v_0)$ 
and a part of the path $\gamma^{top}$ from $(u_0, v_0)$ to $(u^{top},n)$. 
The path $\gamma$ is monotonic, it is contained inside $D_\varepsilon$ and it connects $(u_{bot},0)$ and $(u^{top},n)$.
\end{proof}

According to Lemmas \ref{AltFirst} and \ref{OurSecond} testing condition $\delta(X,Y) \le \varepsilon$ is reduced to finding such $u$ that fulfills
\begin{equation} \label{OurDecision}
(u,0) \in g^\uparrow, \quad (u+m,n) \in g_\downarrow, \quad u+m \le r^\uparrow(u,0),
\quad u \le r_\downarrow(u+m,n).  
\end{equation}
The algorithm presented in the paper by Alt and Godau \cite{frechet} builds a data structure sufficient to test the consistency of conditions (\ref{AltDecision}) using a divide and conquer approach in $O(mn \log(mn))$ time.
Our algorithm is based on condition (\ref{OurDecision}) and therefore uses other data structures that represent functions $r^\uparrow$
and $r_\downarrow$. It is shown that these data structures can be built in two passes in $O(mn)$ time.
Therefore, the subsets $g^\uparrow$, $g_\downarrow$ and functions $r^\uparrow$, $r_\downarrow$ become the main focus of the further considerations.
Section \ref{Achievability} shows that in order to test condition (\ref{OurDecision}) 
it is sufficient to have a finite set of parameters of subsets $g^\uparrow$, $g_\downarrow$ and functions $r^\uparrow$, $r_\downarrow$,
which take $O(mn)$ amount of space.
Section \ref{SimilarityNew} shows that it takes $O(mn)$ time to test condition (\ref{OurDecision}) based on these data.
Finally, Section \ref{InputData} shows that computing these data takes $O(mn)$ time.

\section{\label{Achievability}Formal properties of reachability.}
\begin{lemma} \label{MyMonotony}
For any two points $(u_1, v_1) \in g_\downarrow$ and $(u_2, v_2) \in g_\downarrow$ such that $u_1 \le u_2$ and $v_1 \ge v_2$
the inequality $r_\downarrow(u_1, v_1) \le r_\downarrow(u_2, v_2)$ holds.\\
For any two points $(u_1, v_1) \in g^\uparrow$ and $(u_2, v_2) \in g^\uparrow$ such that $u_1 \le u_2$ and $v_1 \ge v_2$ 
the inequality $r^\uparrow(u_1, v_1) \le r^\uparrow(u_2, v_2)$ holds. 
\end{lemma} %\label{MyMonotony}
\begin{proof}
Let $\gamma_1$ and $\gamma_2$ be two monotonic paths that connect  $(u_1, v_1)$ with $(r_\downarrow(u_1, v_1), 0)$ and $(u_2, v_2)$ with $(r_\downarrow(u_2, v_2),0)$ respectively.
Let us assume that $r_\downarrow(u_1, v_1) > r_\downarrow(u_2, v_2)$. 
As it is shown on Figure~\ref{Fig6} the paths $\gamma_1$ and $\gamma_2$ intersect at some point  $(u_0,v_0)$.
Therefore, a monotonic path that connects $(u_2,v_2)$ with $(r_\downarrow(u_1, v_1), 0)$ exists.
This path consists of a part of the path $\gamma_1$ from $(r_\downarrow(u_1, v_1), 0)$ to $(u_0,v_0)$ and a part of the path $\gamma_2$ from $(u_0,v_0)$ to  $(u_2, v_2)$. 
This means that the point $(u_2,v_2)$ is reachable from a point that is located to the right of the point $(r_\downarrow(u_2, v_2),0)$.
This contradicts with the definition of function $r_\downarrow$. 
Therefore, the assumption $r_\downarrow(u_1, v_1) > r_\downarrow(u_2, v_2)$ is proved to be wrong. The first statement of the theorem is proved. 
The proof of the second statement is similar.
\begin{figure}
  \centering
  \includegraphics{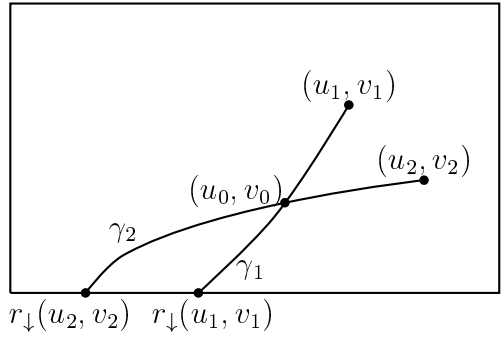}
  \caption{Monotonicity of $r_{\downarrow}$}
  \label{Fig6}
\end{figure}     
\end{proof}
The property of functions $r_\downarrow$ and $r^\uparrow$ stated in Lemma \ref{MyMonotony} will be referred to as monotonicity of these functions. 

For each pair $(i,j)$, $1 \le i \le 2m$, $1 \le j \le n$, denote $D(i,j)$ a square cell 
$$D(i,j) = \{ (u,v)| i-1 \le u \le i, j-1 \le v \le j \}$$
and denote $T(i,j)$, $B(i,j)$, $L(i,j)$ and $R(i,j)$ the top, bottom, left and right sides of this square cell, respectively.
We extend these definitions so that $T(i,0) = B(i,1)$ and $R(0,j) = L(1,j)$.
Denote $D_\varepsilon(i,j) = D_\varepsilon \cap D(i,j)$. It is significant that $D_\varepsilon(i,j)$ is convex for each pair $(i,j)$ \cite{frechet}. 
It is not difficult to prove that for each pair $(i,j)$ 
the intersections $g_\downarrow \cap T(i,j)$, $g_\downarrow \cap R(i,j)$, $g_\uparrow \cap T(i,j)$ and $g_\uparrow \cap R(i,j)$ are also convex, 
that is they are intervals on the sides of the cell $D(i,j)$.
\begin{lemma} \label{Constancy}
For any pair $(i,j)$ the following two statements are valid:\\
if $g_\downarrow \cap R(i,j) \ne \emptyset $ then the function $r_\downarrow$ is constant on $g_\downarrow \cap R(i,j)$;\\
if $g^\uparrow \cap T(i,j) \ne \emptyset $ then the function $r^\uparrow$ is constant on $g^\uparrow \cap T(i,j)$. \phantom{}
\end{lemma}
\begin{proof} 
Let us prove the first statement of the lemma. 
Let $(i,v_1)$ and $(i,v_2)$, $v_1 < v_2$,  be two points from $g_\downarrow \cap R(i,j)$. 
Any point from $B$ that is reachable from $(i,v_1)$ is also reachable from $(i,v_2)$.
Therefore, $r_\downarrow(i,v_1) \le r_\downarrow(i,v_2) $. 
According to Lemma \ref{MyMonotony} on monotonicity of the function $r_\downarrow$ the inequality $r_\downarrow(i,v_1) \ge r_\downarrow(i,v_2) $ is valid as well.
Therefore, $r_\downarrow(i,v_1) = r_\downarrow(i,v_2)$.
The proof of the second statement is similar.
\end{proof}
Restrictions of $r_\downarrow$ to $g_\downarrow \cap T(i,j)$ and $r^\uparrow$ to $g^\uparrow \cap R(i,j)$ are more complex. 
However, they can be presented in a certain standard form.
Let us first take a look at the function $r_\downarrow$ and then extend the result to $r^\uparrow$. 
Let $int \subset g_\downarrow \cap T(i,j)$ be a connected subset called an interval. Let $I$ be a set of pairwise disjoint intervals
such that $g_\downarrow \cap T(i,j) = \bigcup\limits_{int \in I}int$.
\begin{definition}
The restriction of $r_\downarrow$ to $g_\downarrow \cap T(i,j)$ can be expressed in a standard form with a set $I$ of intervals if for each interval $int \in I$ one of the following is valid:\\
-- either $r_\downarrow (u,i) = u$ for all $u \in int$;\\
-- or $r_\downarrow (u_1,i) = r_\downarrow (u_2,i)$ for all $u_1 \in int$, $u_2 \in int$.
\end{definition}
The set $I$  of intervals is ordered so that any nonempty subset has the leftmost and the rightmost intervals (possibly equal). 
The rightmost interval from $I$ is closed. All the other intervals are left-closed (contain the left endpoint) and right-open. 
\begin{lemma}\label{QuantityOfInt}
The restriction of $r_\downarrow$ to $g_\downarrow \cap T(i,j)$ can be expressed in a standard form on a set of no more than $(2j+1)$ intervals.
\end{lemma}
\begin{proof}
Any point $(u,0) \in D_\varepsilon$ is reachable from itself and belongs to $B$. Therefore, $r_\downarrow(u,0)=u$ for all $(u,0) \in g_\downarrow \cap T(i,0)$ and 
the restriction of $r_\downarrow$ to $g_\downarrow \cap T(i,0)$ can be expressed in a standard form with a single interval.

Let us choose an arbitrary $j^*$, $0<j^* \le n$, for the following considerations. 
We will assume that the restriction of $r_\downarrow$ to $g_\downarrow \cap T(i,j^*-1)$ can be expressed in a standard form on a set $I(i,j^*-1)$
of no more than $(2(j^*-1)+1)$ intervals. 
Based on this assumption we will prove that the restriction of $r_\downarrow$ to $g_\downarrow \cap T(i,j^*)$ can be expressed 
in a standard  form on a set $I(i,j^*)$ of no more than $(2j^*+1)$ intervals. 
The proof is partially illustrated by Figure \ref{Fig8}.
\begin{figure}
  \centering
\includegraphics{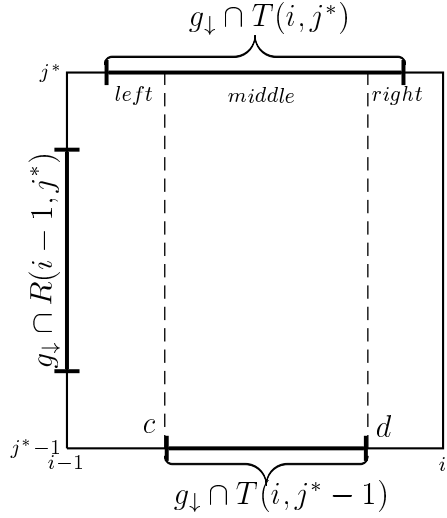}
  \caption{The left, middle and right parts of $g_\downarrow \cap T(i,j^*)$}
  \label{Fig8}
\end{figure}

Denote $c$ and $d$
the horizontal coordinate of the left-most and the right-most points of $g_\downarrow \cap T(i,j^*-1)$ and express $g_\downarrow \cap T(i,j^*)$ 
as the union of three (possibly empty) parts: left, middle and right (see Figure \ref{Fig8}). 
Let us examine the function $r_\downarrow$ on each of these three parts.

The left part consists of points $(u,j^*) \in g_\downarrow \cap T(i,j^*)$, $u < c$.  
Let $(u^*,j^*)$ be one of these points. No point from $(u,j^*-1) \in g_\downarrow \cap T(i,j^*-1)$ is reachable from $(u^*,j^*)$.
Nevertheless, since the point $(u^*,j^*) \in g_\downarrow \cap T(i,j^*)$ exists, $g_\downarrow \cap R(i-1,j^*) \ne \emptyset$ and
since $D_\varepsilon(i,j^*)$ is convex, the point $(u^*,j^*)$ is reachable from any point $(u,v) \in g_\downarrow \cap R(i-1,j^*)$.
According to Lemma \ref{Constancy} the function $r_\downarrow$ is constant on the set $g_\downarrow \cap R(i-1,j^*)$. 
Therefore, the value of the function $r_\downarrow$ on the left part is also constant and is the same as the value of $r_\downarrow$ on $g_\downarrow \cap R(i-1,j^*)$. 
Therefore, as long as the left part exists, the set $I(i,j^*)$ contains an interval that is not present in $I(i,j^*-1)$. This interval is just the left part.  

The middle part consists of points $(u,j^*) \in g_\downarrow \cap T(i,j^*)$, $c \le u \le d$. Let $(u^*,j^*)$ be one of these points. 
Since $D_\varepsilon(i,j^*)$ is convex, the point $(u^*,j^*)$ is reachable from all points $(i,v) \in g_\downarrow \cap R(i-1,j^*)$ 
and from all points $(u,j^*-1) \in g_\downarrow \cap T(i,j^*-1)$, $u \le u^*$. 
Due to monotonicity of the function $r_\downarrow$, it takes its maximum value $r_\downarrow(u^*,j^*-1)$ at the rightmost point $(u^*,j^*-1)$, 
which is also reachable from $(u^*,j^*)$. 
Therefore, $r_\downarrow(u^*,j^*)=r_\downarrow(u^*,j^*-1)$ and this is also valid for any point $(u,j^*)$ on the middle part. 
Since the function $r_\downarrow$ on $g_\downarrow \cap T(i,j^*-1)$ can be expressed in a standard form on a set $I(i,j^*-1)$, 
it can also be expressed on the middle part of the set $g_\downarrow \cap T(i,j^*)$  in a standard form on a set $I(i,j^*)$. 
The set $I(i,j^*)$ is obtained from $I(i,j^*-1)$ by excluding some intervals and changing the endpoints of some other intervals. 
Thus the number of intervals of the middle part does not increase.

The right part consists of points $(u,j^*) \in g_\downarrow \cap T(i,j^*)$, $u \ge d$. 
Let $(u^*,j^*)$ be one of these points. 
The point $(u^*,j^*)$ is reachable from all points $(i,v) \in g_\downarrow  \cap R(i-1,j^*)$ and all points $(u,j^*) \in g_\downarrow \cap T(i,j^*-1)$. 
The function $r_\downarrow$ takes its maximum value $r_\downarrow(d,j^*-1)$ at the rightmost point, which is reachable from $(u^*,j^*)$. 
Therefore, the value $r_\downarrow(u^*,j^*)$ is constant on the right part and equals to $r_\downarrow(d,j^*-1)$. 
Therefore, as long as the right part is not empty, the set $I(i,j^*)$ contains an interval that is not present in $I(i,j^*-1)$. This interval coincides with the right part.

One can see that the number of intervals in $I(i,j)$ can change compared to $I(i,j-\nolinebreak 1)$. 
However, the number of intervals can not increase by more than $2$ intervals. 
These are the intervals that coincide with the left and the right parts of $g_\downarrow \cap T(i,j^*)$.
\end{proof}
Similarly, one can prove the next lemma that we leave without proof.
\begin{lemma}
The restriction of $r^\uparrow$ to $g^\uparrow \cap R(i,j)$ can be expressed in a standard form on a set of no more than $(2i+1)$ intervals.  
\end{lemma}
According to Lemma \ref{Constancy} the set $g^\uparrow \cap B$ and the restriction of $r^\uparrow$ to this set can be expressed with subsets $g^\uparrow \cap T(i,0)$ 
and numbers $r_i^\uparrow$ for $i \in \{1,2, \dots , m \}$. 
The number $r_i^\uparrow$ is the value of $r^\uparrow$ on a subset $g^\uparrow \cap T(i,0)$.
The total amount of these data is of order $m$. 
According to Lemma \ref{QuantityOfInt} the set $g_\downarrow \cap T$ and the restriction of $r_\downarrow$ to this set can be expressed with
the sets $I(i,n)$ of intervals $int$ and with numbers $r_\downarrow^{int}$, where $int \in I(i,n)$, $i \in \{m+1,m+2, \dots , 2m \}$.
Numbers $r_\downarrow^{int}$ define the function $r_\downarrow$ on the interval $int$ in the following way. 
If $r_\downarrow^{int}$ is less than the right endpoint of the interval $int$ then $r_\downarrow(u,n)=r_\downarrow^{int}$ for all $(u,n) \in int$.  
Otherwise, $r_\downarrow(u,n)=u$ for all $(u,n) \in int$.
The total amount of these data is of order $m \times n$.
Denote 
$$C(X, Y) = \Big(\big\langle g^\uparrow \cap T(i,0), r_i^\uparrow, I(i+m,n), r_\downarrow^{int} \big\rangle \Big| i \in \{1,2, \dots , m \}, int \in I(i+m,n) \Big)$$
the data that are sufficient to test conditions (\ref{OurDecision}) for an $m$-gonal curve $X$ and an $n$-gonal curve $Y$.
Section \ref{SimilarityNew} shows that it takes $O(m n)$ time to test condition (\ref{OurDecision}) based on these data. 
Section \ref{InputData} shows how to obtain the data $C(X, Y)$ in $O(m n)$ time.     

\section{Testing the similarity of closed polygonal curves}\label{SimilarityNew}
\begin{lemma}\label{TestingOfMainInequality}
If for closed polygonal curves $X$ and $Y$ with $m$ and $n$ vertices the data $C(X, Y)$ are known then testing $\delta(X,Y) \le \varepsilon$ can be done in $O(m n)$ time.
\end{lemma}
\begin{proof} According to Lemmas \ref{AltFirst} and \ref{OurSecond}, condition $\delta(X,Y) \le \varepsilon$ is equivalent to the existence of a number $u$, $0 \le u \le m$, 
that fulfills conditions
\begin{equation} \label{OurDecisionCopy}
(u,0) \in g^\uparrow, \quad (u+m,n) \in g_\downarrow, \quad u+m \le r^\uparrow(u,0),
\quad u \le r_\downarrow(u+m,n).  
\end{equation}
Since the data $C(X,Y)$ are known, testing condition (\ref{OurDecisionCopy}) is reduced to testing if there exists a triple 
$i \in \{ 1,2, \dots , m \}$, $int \in I(i+m,n)$, $u \in \{t|0 \le t \le m\}$ that fulfills conditions
\begin{equation} \label{OurDecisionConcreteNew }
(u,0) \in g^\uparrow \cap T(i,0), \quad (u+m,n) \in int, \quad u+m \le r_i^\uparrow,
\quad u \le r_\downarrow(u+m,n).  
\end{equation}
Let us replace the condition $u \le r_\downarrow(u+m,n)$ in (\ref{OurDecisionConcreteNew }) with a condition $u \le r_\downarrow^{int}$
and express (\ref{OurDecisionConcreteNew }) in a form
\begin{equation} \label{OurDecisionAuxiliary}
(u,0) \in g^\uparrow \cap T(i,0), \quad (u+m,n) \in int, \quad u+m \le r_i^\uparrow,
\quad u \le r_\downarrow^{int}.
\end{equation}
Conditions (\ref{OurDecisionConcreteNew }) and (\ref{OurDecisionAuxiliary}) are equivalent.
Indeed, if the function $r_\downarrow$ is constant on $int$ then $r_\downarrow(u+m,n) = r_\downarrow^{int}$ for all $(u,n) \in int$. 
If $r_\downarrow(u,n)=u$ on $int$ then condition $u \le r_\downarrow(u+m,n)$ becomes an inequality $u \le u+m$ that is valid for all $u$. 
In this case, the value $r_\downarrow^{int}$ equals to a horizontal coordinate of interval's $int$ right endpoint. 
Therefore, the inequality $u \le r_\downarrow^{int}$ is also valid for any point $(u,n) \in int$.

Denote $a^{int}$ and $b^{int}$ the left and right endpoints of the interval $int$.
Denote  $c_{i}$ and  $d_{i}$ the horizontal coordinate of the leftmost and the rightmost points of $g^\uparrow \cap T(i,0)$ respectively. 
Using this notation, condition $(u,0) \in g^\uparrow \cap T(i,0)$ in (\ref{OurDecisionAuxiliary}) becomes $c_{i} \le u \le d_{i}$. 
Condition $(u+m,n) \in int$ becomes $a^{int}-m \le u \le b^{int}-m$ when $int$ is the rightmost interval in $I(i+m,n)$ and becomes $a^{int}-m \le u < b^{int}-m$ otherwise. 
Therefore, for some triples the constraint
\begin{equation} \label{OurDecisionWeekConcrete2 }
c_{i} \le u \le d_{i}, \quad a^{int}-m \le u \le b^{int}-m, \quad u \le r_{i}^\uparrow - m,
\quad u \le r_\downarrow^{int}
\end{equation}
is weaker than (\ref{OurDecisionAuxiliary}). Nevertheless, if a triple
$(i, int, u)$ that fulfills (\ref{OurDecisionWeekConcrete2 }) exists then the  triple that fulfills (\ref{OurDecisionAuxiliary}) also exists.
Let $(i^*, int^*, u^*)$ be a triple that fulfills (\ref{OurDecisionWeekConcrete2 }) and does not fulfill (\ref{OurDecisionAuxiliary}).
In other words, $a^{int^*}-m \le u^* \le b^{int^*}-m$ and $ (u+m,n) \notin int^*$. 
It is only possible when $u^* = b^{int^*}-m$ and $int^*$  is not the rightmost interval from $I(i^*+m,n)$. 
It follows that an interval $int^+ \in I(i^*+m,n)$ exists such that $a^{int^+} =~b^{int^*}$. 
The triple $(i^*, int^+, u^*)$ fulfills conditions (\ref{OurDecisionAuxiliary}). 
Therefore, the consistency of conditions (\ref{OurDecisionWeekConcrete2 }) is equivalent to the consistency of conditions (\ref{OurDecisionAuxiliary}), although constraints (\ref{OurDecisionWeekConcrete2 }) are weaker than (\ref{OurDecisionAuxiliary}). 
In turn, consistency of (\ref{OurDecisionWeekConcrete2 }) is equivalent to the existence of a pair
$i \in \{ 1,2, \dots , m \}$, $int \in I(i+m,n)$,  that fulfills the inequality
%\begin{equation} \label{OurDecisionMolecular1New}
$$\max \{a_i, \quad a^{int}-m  \} \le 
\min \{ d_i,\quad b^{int}-m,\quad r_i^\uparrow - m,\quad r_\downarrow^{int} \}.$$ 
%\end{equation}
Testing this inequality for any pair $(i,int)$ can be performed in constant time. 
According to Lemma \ref{QuantityOfInt} the number of intervals tested for each $i$ does not exceed $(2n+1)$. Consequently, the number of tested pairs $(i,int)$ is $O(m n)$. 
\end{proof}

\section{Obtaining the data $C(X, Y)$}\label{InputData}
\subsection{The general scheme}
Let the intervals $D_\varepsilon  \cap T(i,j)$, $0 \le i \le 2m$,  $1 \le j \le n$, and $D_\varepsilon  \cap R(i,j) $, $1 \le i \le 2m$,  $0 \le j \le n$,
be built for closed polygonal curves $X$ and $Y$.
Based on these data it is necessary to build $C(X,Y)$, which is used to test curve similarity as described in Section \ref{SimilarityNew}. 
The data $C(X,Y)$ consist of sets $g_\downarrow \cap T(i,n)$, $m+1 \le i \le 2m$, restrictions of $r_\downarrow$ to these sets, 
sets $g^\uparrow \cap T(i,0)$, $ 1 \le  i \le m $, and values of function $r^\uparrow$ on these sets. 
These sets and functions are built in two independent stages that we call forward and backward passes.
Both passes consist of $2mn$ steps, pair $(i,j)$ being the number of the step.

A forward pass performs a step number $(i,j)$ after steps $(i-1,j)$ and $(i,j-1)$.
On a step number $(i,j)$ the subsets $g_\downarrow \cap R(i,j) $, $g_\downarrow \cap T(i,j)$ 
and restrictions of $r_\downarrow $ to these subsets are built.
These data are obtained based on $g_\downarrow \cap R(i-1,j) $, 
$g_\downarrow \cap T(i,j-1) $ and restrictions of $r_\downarrow $ to these sets, which had been built on previous steps.
The result of all $2mn$ steps of a forward pass are the sets $g_\downarrow \cap T(i,n)$, $m+1 \le i \le 2m $, 
and restrictions of $r_\downarrow$ to these sets. 
The proof of Lemma \ref{QuantityOfInt} practically describes the idea of an algorithm for the forward pass. 
Nevertheless, Subsection \ref{Forward} describes it in more detail.

The result of the forward pass is only one part of the data needed to test (\ref{OurDecisionConcreteNew }). 
Another part is the result of the backward pass.
During the backward pass a step number $(i,j)$ is performed after steps number $(i+1,j)$ and $(i,j+1)$. 
A step number $(i,j)$ builds subsets $g^\uparrow \cap L(i,j) $, $g^\uparrow \cap B(i,j)$ 
and restrictions $r^\uparrow $ to these subsets based on subsets $g^\uparrow \cap L(i+1,j) $, $g^\uparrow \cap B(i,j+1) $ and restrictions of $r^\uparrow $ to these subsets. 
The result of the backward pass are subsets $g^\uparrow \cap T(i,0)$ and values of $r^\uparrow$ on these subsets.
The backward pass uses the same approach as the forward pass and therefore is not presented in detail by this paper.
   
\subsection{The forward pass}\label{Forward}
We define a data structure to store any function $f$ that can be expressed in a standard form on a set $I$ of intervals $int$. 
The function $f$ on each of these intervals $int \in I$ is defined by a triple $(beg, val, end )$, numbers $beg$ and $end$ being the endpoints of the interval $int$. 
If $val < end$ then $f(x)=val$ for all $x \in int$, otherwise if $val = end$ then $f(x)=x$ for all $x \in int$. 
The triples $(beg, val, end )$ are stored in a double-ended queue (deque) $Q$ with the following operations performing in constant time:\\
-- testing whether the deque is empty;\\
-- reading and removing either the leftmost or the rightmost triple;\\
-- inserting a triple either to the left or to the right end of the deque.\\
For any given number $x$ we define additional operations of cutting the deque to the left of $x$ and cutting it to the right of $x$.
When the deque is being cut to the left of $x$ all such triples $(beg, val, end)$ that $end < x$ are removed from the left end of the deque. 
Then if the triple $(beg,val,end)$ on the left end is such that $beg \le x < end$, it gets replaced with a triple $(x,val,end)$.
When the deque is being cut to the right of $x$ all such triples $(beg, val, end)$ that $beg > x$ are removed from the right.
Then the triple $(beg,val,end)$ on the right end with $beg \le x \le end$ is replaced either by a triple $(beg,val,x)$ if $val < end$, or by a triple $(beg,x,x)$ if $val = end$. 
Since the triples in the deque are sorted, the time spent on cutting the deque is proportional to the number of triples removed from the deque.

The input data for the forward pass are \\
\phantom{} \quad $D_\varepsilon \cap T(i,j)$, $1 \le i \le 2m$, $0 \le j \le n$,\quad and \quad $D_\varepsilon \cap R(i,j)$, $0 \le i \le 2m$, $1 \le j \le n$.\\
The algorithm works with $2m$ deques $Q(i)$,  $i \in \{1, 2, \dots , 2m \}$.
Each $(i,j)$-th step starts with a deque $Q(i)$ that represents the set $g_\downarrow \cap T(i,j-1)$ and
the restriction of $r_\downarrow$ to this set
and updates it to represent the set $g_\downarrow \cap T(i,j)$ and
the restriction of $r_\downarrow$ to this set.
Thus, after the algorithm finishes the deque $Q(i)$ represents the set $g_\downarrow \cap T(i,n)$ and the restriction of $r_\downarrow$ to this set.
In order to perform this update on each step $(i,j)$ the auxiliary data in the form of the subset $g_\downarrow \cap R(i,j-1)$
and the value of $r_\downarrow$ on this subset is needed.
On a step number $(i, j)$ the subset $g_\downarrow \cap R(i,j)$ and the value of $r_\downarrow$ on this subset is calculated
and the deque $Q(i)$ is updated.

The algorithm consists of an initialization and $2mn$ steps.
The initialization sets the initial states of the deques and the initial values of the auxiliary data according to the following rules.

If $(0,0) \in  D_\varepsilon$ then $g_\downarrow \cap R(0,1)=D_\varepsilon \cap R(0,1)$ and the function $r_\downarrow $ takes value $0$ on this set. 
Otherwise $g_\downarrow \cap R(0,1) = \emptyset$. 

If $(0,j-1) \in g_\downarrow$ then $g_\downarrow \cap R(0,j)=D_\varepsilon \cap R(0,j)$ and the function
$r_\downarrow $ takes value $0$ on this set. 
Otherwise $g_\downarrow \cap R(0,j) = \emptyset$. 

The set $g_\downarrow \cap T(i,0)$, $1 \le i \le 2m$, equals $D_\varepsilon \cap T(i,0)$.  
If $D_\varepsilon \cap T(i,0) \ne \emptyset $ then the function $r_\downarrow $ takes value $r_\downarrow (u,0) = u$ on this set. 
In this case the deque $Q(i)$ has a single triple $(beg, val, end)$, number $beg$ being the horizontal coordinate of the leftmost point of $D_\varepsilon \cap T(i,0)$
and $val=end$ being the coordinate the rightmost point. 
If $D_\varepsilon \cap T(i,0) = \emptyset $ then the deque $Q(i)$ is initially empty. 
It takes $O(m+n)$ time to perform the initialization. 

At the beginning of a step number $(i,j)$ the following data are known.
The set $g_\downarrow \cap R(i-1,j)$ and the value of function $r_\downarrow$ on this set that we denote $r_\downarrow^*$.
The subset $g_\downarrow \cap T(i,j-1)$ and the restriction of $r_\downarrow$ to this subset that are represented by a deque $Q(i)$.
These data are illustrated on Figure~\ref{Fig9}. 
The picture also shows the interval $D_\varepsilon \cap T(i,j)$, which is given as input to the forward pass, 
numbers $a$ and $b$ being the left and right endpoints of this interval, $c$ and $d$ being the endpoints of the interval $g_\downarrow \cap T(i,j-1)$.
On each step $(i,j)$ the algorithm computes the set $g_\downarrow \cap R(i,j)$ and the value of $r_\downarrow$ on this set and updates the deque $Q(i)$ to
represent the subset $g_\downarrow \cap T(i,j)$ and the restriction of $r_\downarrow$ to this subset.
\begin{figure}
  \centering
  \includegraphics{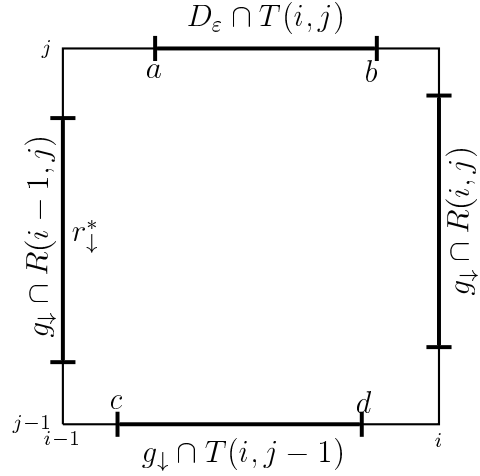}
  \caption{The data processed on $(i, j)$-th step}
  \label{Fig9}
\end{figure}

Computing the set $g_\downarrow \cap R(i,j)$ and the value of $r_\downarrow$ on this set can be done by a rather simple rule in constant time,
which is not described here.
Updating the deque $Q(i)$ is a bit more complex but still consists of several rather straightforward rules.
\\
1. If $D_\varepsilon \cap T(i,j) = \emptyset$ then remove all triples from $Q(i)$.\\
2. If $g_\downarrow \cap R(i-1,j)= \emptyset$ and $g_\downarrow \cap T(i,j-1) = \emptyset$ then the  deque $Q(i)$ is empty and it should remain empty.\\
3. If $D_\varepsilon \cap T(i,j) \ne \emptyset$, $g_\downarrow \cap R(i-1,j)\ne \emptyset$ and
$g_\downarrow \cap T(i,j-1) = \emptyset$ then the deque $Q(i)$ was empty at the beginning of the step $(i ,j)$. 
Since $D_\varepsilon \cap D(i,j)$ is convex, any point of $D_\varepsilon \cap T(i,j)$ is reachable from any point of  $g_\downarrow \cap R(i-1,j)$.
Therefore, a single triple $(a,r_\downarrow^* , b)$ is inserted into the empty deque $Q(i)$.

If $D_\varepsilon \cap T(i,j) \ne \emptyset$ and $g_\downarrow \cap T(i,j-1) \ne \emptyset$ then the update of $Q(i)$ 
depends on the relative position of intervals $D_\varepsilon \cap T(i,j)$ and $g_\downarrow \cap T(i,j-1)$. 
The different cases of relative positions are shown on Figure~\ref{Fig11}.\\
4. Case a): exclude all triples from $Q(i)$; if $g_\downarrow \cap R(i-1,j)\ne \emptyset$ then insert the triple $(a,r_\downarrow^* , b)$
is inserted into the empty deque $Q(i)$.\\ 
5. Case b): cut the deque to the right of $b$; if $g_\downarrow \cap R(i-1,j)\ne \emptyset$ then insert the triple $(a,r_\downarrow^* , c)$ to the left end of $Q(i)$.\\ 
6. Case c): if $g_\downarrow \cap R(i-1,j)\ne \emptyset$ then insert $(a,r_\downarrow^* , c)$ to the left end of $Q(i)$; 
    insert $(d, r_\downarrow(d,0), b)$ to the right end of the deque $Q(i)$. \\
7. Case d): cut the deque to the left of $a$; cut the deque to the right of $b$. \\
8. Case e): cut the deque to the left of $a$; insert $(d, r_\downarrow(d,0), b)$ to the right end of the deque $Q(i)$. \\
9. Case f): remove all triples from $Q(i)$; insert a single triple $(a, r_\downarrow(d,0), b)$ to $Q(i)$.
\begin{figure}
  \centering
  \includegraphics{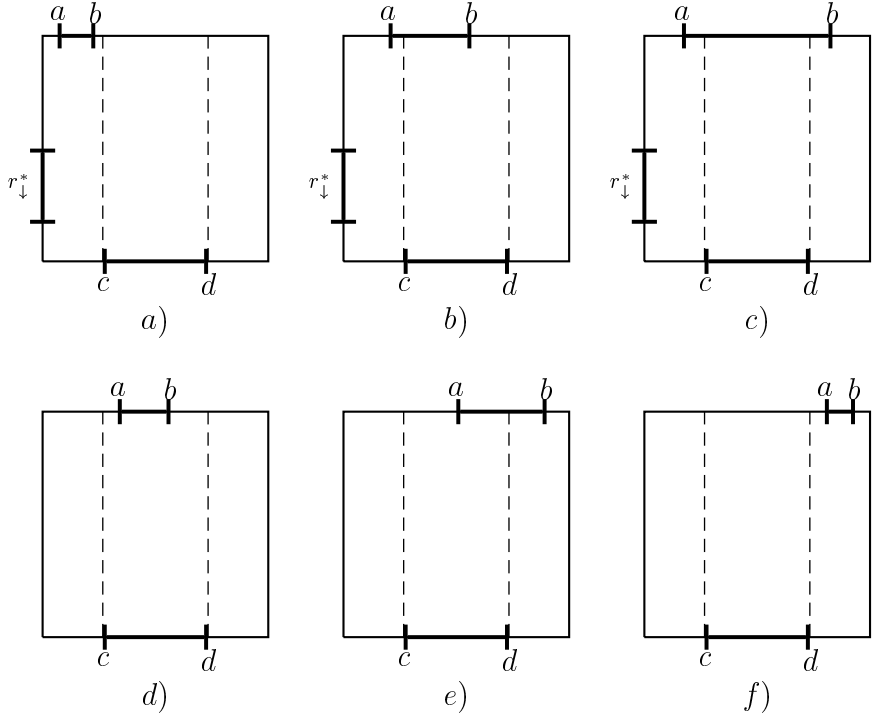}
  \caption{Relative position of intervals $D_\varepsilon \cap T(i,j)$ and $g_\downarrow \cap T(i,j-1)$}
  \label{Fig11}
\end{figure}  
\begin{lemma}\label{ComplexityForwardAndBackward}
It takes $O(m n)$ time to complete both the forward and the backward passes of the algorithm.
\end{lemma}
\begin{proof}
We prove this lemma for the forward pass. The proof for the backward pass is similar.\\
No more than two triples are inserted into the deques on each step. Therefore, no more than $4mn$ triples are inserted during all steps $(i,j)$, $1 \le i \le 2m$, $1 \le j \le n$.
The number of triples removed from the deques does not exceed the number of triples inserted into the deques, and therefore is not greater than $4mn$.
Every time some triple is read from the deques it is also removed. Therefore, the number of times the triples are read is also not greater than $4mn$.

Therefore, the forward pass includes the initialization, which takes $O(m+n)$ time, and the work with deques, which takes $O(m n)$ time.
In addition, the forward pass also computes the set $g_\downarrow \cap R(i,j)$ and the value of $r_\downarrow$ on this set on each step.
It takes constant time on each step to do this. Therefore, the total amount of time spent on these computations is $O(m n)$.
\end{proof}              
\section{The result}\label{ResultSection}
\begin{theorem} 
Let $X$ and $Y$ be closed polygonal curves with $m$ and $n$ vertices and $\delta(X,Y)$ be the Frechet distance between them. 
Testing the inequality $\delta(X,Y) \le \varepsilon$ takes $O(m n)$ time.
\end{theorem}
\begin{proof}
Testing the inequality $\delta(X,Y) \le \varepsilon$ is reduced to the following computations.
\\
For the given polygonal curves $X$ and $Y$ the sets $D_\varepsilon \cap T(i,j)$, $1 \le i \le 2m$, $0 \le j \le n$,  and $D_\varepsilon \cap R(i,j)$, $0 \le i \le 2m$, $1 \le j \le n$, are computed.
This is equivalent to solving $2mn$ quadratic equations.
\\
Then based on the sets $D_\varepsilon \cap T(i,j)$ and $D_\varepsilon \cap R(i,j)$ the subsets $g_\downarrow \cap T(i,n)$, $m+1 \le i \le 2m$, 
and restrictions of $r_\downarrow$ to these subsets are computed along with the subsets $g^\uparrow \cap T(i,0)$, $ 1 \le  i \le m $, and the values of $r^\uparrow$ on these subsets. 
According to Lemma \ref{ComplexityForwardAndBackward} it takes $O(m n)$ time to build these.
\\
Finally, testing the inequality $\delta(X,Y) \le \varepsilon$ based on these data can be done in $O(m n)$ time according to Lemma \ref{TestingOfMainInequality}.
\end{proof}

\end{document}